\newcommand*{\LONG}{}%
\pdfoutput=1

\ifdefined\LONG
\documentclass[11pt]{article}
\usepackage{amsfonts,amsmath,amssymb,latexsym,comment,amsthm}

\usepackage[letterpaper]{geometry}
\usepackage{fullpage}

\usepackage{mathrsfs}
\usepackage{times}

\else 
\documentclass{sig-alternate-2013}
\usepackage{amsfonts,amsmath,amssymb,latexsym,comment }

\fi

\usepackage{qtree}
\usepackage{epsfig}
\usepackage{graphicx}
\usepackage{subfigure}

\usepackage{xspace}
\usepackage{cases}
\ifx\pdftexversion\undefined
\usepackage[colorlinks,linkcolor=black,filecolor=black,citecolor=black,urlcolor=black,pdfstartview=FitH]{hyperref}
\else
 \usepackage[colorlinks,linkcolor=blue,filecolor=blue,citecolor=blue,urlcolor=blue,pdfstartview=FitH]{hyperref}
\fi
\usepackage[lined,ruled,boxed,commentsnumbered,noresetcount]{algorithm2e}
\usepackage{enumitem}

\usepackage{pifont}

\def\squarebox#1{\hbox to #1{\hfill\vbox to #1{\vfill}}}

\newtheorem{theorem}{Theorem}

\newtheorem{definition}{Definition}

\newtheorem{claim}{Claim}

\newcommand{\namedref}[2]{\hyperref[#2]{#1~\ref*{#2}}}
\newcommand{\sectionref}[1]{\namedref{Section}{#1}}

\newcommand{\figureref}[1]{\namedref{Figure}{#1}}

\newcommand{\claimref}[1]{\namedref{Claim}{#1}}

\newcommand{\lref}[1]{\namedref{Line}{#1}}

\newcommand{\true}{\mathit{true}}

\newcommand{\ASMfm}{\mbox{ASM\!$f\!m$}\xspace}
\newcommand{\SMPfm}{\mbox{SMP\!$f\!m$}\xspace}
\newcommand{\SMPm}{\mbox{SMP$m$}\xspace}
\newcommand{\SBMPfm}{\mbox{SBMP\!$f\!m$}\xspace}

\newcommand{\tri}[1]{\left<#1 \right>}

\def\beginsmall#1{\vspace{-\parskip}\begin{#1}\itemsep-\parskip}
\def\endsmall#1{\end{#1}\vspace{-\parskip}}

\newcommand{\kSetCons}{\textsc{Dynamic $k$-vector-set Consensus}\xspace}
\newcommand{\mSetCons}{\textsc{Dynamic $(m+1)$-vector-set Consensus}\xspace}
\newcommand{\mAD}{{\textit{m-}\textsc{ad}}\xspace}
\newcommand{\oAD}{{{\small 1-}\textsc{ad}}\xspace}

\newcounter{linenumbers}

\newcounter{todocounter}
\newcommand{\todonum}
{\stepcounter{todocounter}{(\thetodocounter)}}

\newcommand{\dd}[1]{\textbf{\color{blue}
[[\todonum\ dd: #1]]}}

\newcommand{\tb}{\makebox[0.6cm]{}}
\newcommand{\due}{\makebox[1cm]{}}

\newcommand{\hide}[1]{}

\newcommand{\commentout}[1]{}

\newcommand{\Byzantine}[0]{\mbox{\emph{Byzantin}\hspace{-0.15em}\emph{e}}\xspace}

\newcommand{\ca}{{\small\textsc{commit\!\_adopt}}\xspace}

\newcounter{linecounter}
\newcommand{\linenumbering}{\ifthenelse{\value{linecounter}<10}
{(0\arabic{linecounter})}{(\arabic{linecounter})}}

\renewcommand{\thelinecounter}{\ifnum \value{linecounter} > 
9\else 0\fi \arabic{linecounter}}

\begin{document}

\title{Synchronous Hybrid Message-Adversary}
\author{Danny Dolev (HUJI)\footnote{Email: danny.dolev@mail.huji.ac.il.} \ and
Eli Gafni (UCLA)\footnote{Email: gafnieli@gmail.com.}\\[1.2ex]
regular submission, not eligible for the best student paper award\vspace{-2ex}}

\date{}

\maketitle

\thispagestyle{empty}

\setcounter{page}{0}

\section*{Abstract}
The theory of distributed computing, lagging in its development behind practice, has been biased in its modelling by  employing mechanisms within the model mimicking reality. Reality means, processors can fail.
But theory is about predicting consequences of reality, hence if we capture reality by ``artificial models,'' but those nevertheless make analysis simpler, we should
pursue the artificial models.

Recently the idea was advocated to analyze distributed systems and view processors as infallible. It is the message delivery substrate that causes problems.
This view not only can effectively emulate reality, but above all seems to allow to view any past models as 
\emph{synchronous} 
models. Synchronous models are easier to analyze than asynchronous ones. Furthermore, it gives rise to models we haven't contemplated in the past. One such model, presented here, is the Hybrid Message-Adversary. We motivate this model through the need to analyze Byzantine faults. The Hybrid model exhibits a phenomenon not seen in the past.

\newpage

\section{Introduction}
Recently, in \cite{AG15}, it was shown that a read-write wait-free asynchronous system can be modelled by a synchronous
message passing system
with a message adversary. This was a step beyond \cite{ABD90,ABD95} that equated shared-memory and message-passing when at most a minority of the
processors are faulty. The work in \cite{AG15}, 
not only equated message-passing and shared-memory for the potential of $n-1$ faults,
but in addition the message-passing system was synchronous. And, it paid handsomely with an almost trivial derivation of the 
necessity condition of asynchronous computability \cite{HS}.

Here, as a first step, we extend the study of \cite{AG15} to the $m$-resilient read-write asynchronous system,
and find the message adversary \mAD in the synchronous message passing model that makes the synchronous system equivalent to the asynchronous one.
\mAD is a system in which in each synchronous round \mAD chooses $m$ processors and can remove any of the messages they send, subject to the condition that for the two messages exchanged
between a pair of processors from the chosen $m$ processors in the round, it removes at most one message.

This follows from the simple realization that in \cite{AG15}, when we take a system of two processors, instead of viewing
the adversary as one that at each round can remove one message of the two sent, we could equivalently say that the adversary chooses a processor and can remove
the message it sends. A simple semantic difference, that was not realized by us for a while.

The structure of \oAD was realized 3 decades ago in \cite{Healer} but not pursued beyond $m=1$.
The celebrated transformation of shared-memory to message-passing in \cite{ABD90,ABD95}, which holds only for $m<(1/2)n$,
is hard to cast as a message-adversary that chooses a set of processors and removes some of its messages.
It is more than a decade later that the second author quoted in \cite{983102} discovered that after 3 rounds of the transformed
system it can be viewed as a message-adversary.

When solving colorless tasks (\cite{BGLR}), if $m<(1/2)n$, we can do away with this last condition that prevents \mAD from removing the two messages exchanged by faulty processors. 
For the rest of the paper we consider only colorless tasks, hence we allow \mAD a complete freedom in removing in a round any messages from the $m$ processors it chose for the round. We call this type of adversary \emph{$m$-mobile omission faults}.

This paper deals with just a single colorless task. The colorless task we consider  is the binary-vector-consensus-task.
To explain this task we first put it within the historical perspective of the development of the various versions of the set-consensus task.

The ingenious original version of set-consensus task invented by Chaudhuri \cite{SC}, was
a version we call here, \emph{set-election}: $n$ processors start each with their own identifier (id) as its input,
and each outputs a participating id, such that the number of distinct outputs should be smaller than $n$.

Using topological arguments this problem was later proved to be unsolvable read-write wait-free \cite{BG,HS,SZ}.
About a decade and a half later a technically simple but profound
paper \cite{VSC} formulated the task of \emph{vector-set-consensus}. In the vector-set-consensus task,
the $n$ processors are faced with $n-1$ independent consensuses. Each processor outputs for at least one of the consensuses,
such that, two processors that output for the same consensus satisfy validity and agreement for that consensus.
A simple reduction to set-election showed the vector-set-consensus to be unsolvable read-write wait-free,
and to be equivalent to set election.

But maybe, had we restricted each of the consensuses in the vector-set-consensus task to be just binary-consensuses,
rather than multi-valued consensus, then the resulting binary-vector-set-consensus is solvable
read-write wait-free? 
Precluding this possibility is equivalent to proving Generalized Universality \cite{Rachid},
a major result, which succeeded only few years after the formulation of the vector-set-consensus task.  

The various tasks above have the obvious analogue when we change $n-1$ to $k < n-1$.
Indeed, Chaudhuri \cite{SC} proposed the problem in the context of $m$-resilient system, and asked
whether the system can solve $m$-set consensus. The BG-simulation in \cite{BG,BGLR} showed that had her problem
been solvable read-write $m$-resiliently, then set-election could have been solved for $m+1$ processors read-write
wait-free, which is impossible. On the flip-side, a $m$-resilient system can trivially solve $m+1$ election! 

But all these analyses were conducted in the context of benign failures. What if we have an asynchronous
system with $n$ processors and $f<(1/3)n$ Byzantine failures? Obviously, Byzantine failures are more serious failures than
omission failures, and since we cannot solve $f$-binary-vector-set-consensus in the $n$ processors $f$-resilient
model we cannot solve it in the asynchronous $n$ processors system with $f$ Byzantine faults.
Although the Byzantine failures are on a fixed set of processors, if these processors just lie about say 
their inputs, they cannot be detected, and a processor has to move on after receiving $n-f$ messages,
lest the $f$ it missed are the Byzantine set, which omitted messages. 
Hence $f+1$-set consensus is a lower bound in this case too. But can we solve
the $f+1$-vector-binary-set-consensus problem, in $n$ processors system with $f<(1/3)n$ Byzantine faults like we could do
it for the analogue $f$-resilient system? 

To our knowledge it is the first time this question has been asked.
This is not surprising since the notion of binary-vector-set-consensus can only really be understood on the background of \cite{Rachid}.
The question was asked with respect to the election version \cite{Dalia}, but not for the vector-set-consensus task, let alone the binary-vector-set consensus task, since these tasks
had not been introduced as yet, then.
Why wasn't it asked with respect to the vector-set-consensus? Since these are reducible to each other without
introducing the ideas in \cite{Rachid}. It looked meaningless to recast the results in \cite{Dalia} for vector-set-consensus. We do not believe that the general multi-valued set consensus task with 
its 24 (at the time!)  possible versions, or for that matter even the general vector-set-consensus has any bearing on the vector-binary-set-consensus,
when we consider Byzantine faults. Hence we do not pursue this multi-valued route, since we do not know a reasonable formulation of it in the Byzantine case.
Thus, we are left with the question of the vector-binary-set-consensus power of
asynchronous Byzantine $n,f$ system.


We thus investigate this new general type of system. A \emph{synchronous} system with $f$-fixed Byzantine faults and in addition
$m$-mobile omission faults. We show this system to be equivalent in its vector-binary-set-consensus power
 to the same system of less 
severe failure of omission rather than Byzantine. We therefore investigate the general model of $f$-fixed omission faults
and $m$-mobile omission faults. We call such a system that intermingles fixed omission failure and mobile 
omission failure, an \emph{hybrid adversary}.

We completely characterize the vector-set-consensus power of the synchronous hybrid fix/mobile omission system with
respect to all set-consensus versions, since
in the non-hybrid benign faults all these versions are equivalent.

We do  not show a reduction between the two synchronous systems one with the $f$-fixed Byzantine failures and $m$-mobile omissions, and the other
with both types of failures being omission failures. Byzantine processors can always alter their input
and behave correctly there on. Thus, when inputs conflict there is no resolution as to who are the ``good ones'' and 
who are the ``bad ones.'' Rather, we should consider only contexts in which such resolution is not needed. To show
such a result, the way  we speculate to do it is first to show that the two systems have the same vector-binary-set-consensus power,
which we do here.

Obviously the set consensus power of the benign system with fixed omissions and mobile omissions is stronger
than the analogue synchronous hybrid Byzantine system. Thus a lower bound to the former
is a lower bound to the latter.  To show equivalence we present an algorithm
for the Byzantine system that gets the same consensus power as that of the benign system, explicitly.

Our way of proceeding first with the binary-vector-set-consensus-power is in line with our recent thinking that systems
that agree on set consensus tasks have the same power when other tasks are concerned. That is the thinking of:
``set consensus tasks are the coordinates of any (reasonable) system,'' \cite{CSC}.

We obtain the result that an $n,f,m$ omission system requires $n \geq (m+1)f+1$ in order to solve the best value of 
set-consensus it can solve, $m+1$. For lesser values of $n$ it can solve set consensus $m+j$, when 
$n \geq \lfloor f(m+j)/j \rfloor +1$.

This is the first time that we see a model where its set-consensus power changes gradually with $n$.
In the $m$-resilient shared-memory case, we can always solve $m+1$-set consensus. In the message
passing system when $m > n/2-1$ we cannot solve anything. Once we are below that threshold we can
solve $m+1$-set consensus. It is the combination of faults that give rise to this gradually changing
power phenomenon.

\subsection{Outline of the Paper}
We first show and prove the synchronous analogue of the asynchronous $m$-resilient model,
and show that some restrictions on the adversary can be removed in the case of colorless-tasks (\cite{BGLR}).

We then characterize the hybrid $n,f,m$ omission adversary for different combinations of these values.
We use the  \cite{BG} simulation to show that  if a set of less than $(m+1)f+1$ processors can do $m+1$ set consensus then 
$m+2$ processors can do set-consensus wait-free contradicting \cite{BG,HS,SZ}. 

Finally we show that the upper bound algorithm holds for the hybrid synchronous Byzantine system
when dealing with binary-vector-set-consensus. 
For the upper bound we use the rotating coordinator algorithm \cite{Reischuk198523,Chandra:1991:TME:112007.112027}.


\section{Problem Statement and Models}
We assume a set of $n$ processors 
$\Pi=\{p_1, p_2, ..., p_n\}$. 
The paper focuses on solving the $k$-vector-set consensus problem in the hybrid omission model $\SMPfm$, to be explained below.

\begin{definition}[$k$-vector-set consensus problem]
Every processor has a vector of $k$ initial inputs. Each processor, $p$, returns a vector $o_p[1..k]$ that satisfies:
\beginsmall{enumerate}
\item[\textup{ vset1:}] If for any index $j\in[1..k]$ the inputs of all processors are the same, then every processor returns  that value for index $j$.
\item[\textup{ vset2:}] For every entry $j\in[1..k]$ and for every two processors, $p$, and  $q$,  if $o_p[j]\not=\bot$ and $o_q[j]\not=\bot$ then $o_p[j]=o_q[j]$,
and it is one of the inputs at entry $j$ for some participating processor.
\item[\textup{ vset3:}] There exists an index $j\in[1..k]$ such that for every two processors, $p$, and  $q$,   $o_p[j]=o_q[j]$.
 \endsmall{enumerate}
\end{definition}

Notice that requiring all to output for the \emph{same} index $j$, is equivalent to asking just to output for some $j$, since processors
can write their outputs, and then read and adopt outputs it sees. The first output written will be adopted by all.

\begin{definition}[binary-$k$-vector-set consensus problem]
Same as the above only that the vector of $k$ initial inputs each processor has is a binary vector of 0's and 1's.
\end{definition}


The model \SMPfm, called \emph{hybrid omission $n,f,m$}, is a synchronous point-to-point message passing system.
We consider an adversary that can remove messages.
Before the start of the algorithm the adversary chooses a set $S^f$ of $f$ processors.
In each round the adversary can remove some or all of the messages sent by $S^f$.
In addition in each synchronous round it can choose a set $S^m$ of $m$ processors and remove some 
or all of the messages sent from $S^m$.
We call $S^f$, the fixed set, and $S^m$ the mobile set.

Presenting models as ``message adversary'' enables us to easily deal with dynamic systems in which processors that  presented an external erroneous behavior at one point to start behaving correctly later on, without the need to discuss what is their internal state when this happens.

Where is \SMPfm coming from.
The method in \cite{AG15} transformed asynchronous wait-free SWMR shared memory into a synchronous message-passing adversary.
This message passing adversary can be viewed, in hindsight, as an adversary that at each round chooses $n-1$ processors
and removes some of their messages so that between two processors it chose, at least one of the two messages sent between
them is not removed.

In the asynchronous $m$-resilient SWMR iterated shared memory, which is equivalent to the classical $m$-resilient 
non-iterated model \cite{BG97,tasks,HR10}, 
at each iteration processors are presented with a fresh SWMR memory initialized to $\bot$,
and all write their cell and then snapshot the memory. The $m$-resiliency assumption entails that for each processor
the snapshot returns at most $m$ cells with the value $\bot$.

We now claim this model is equivalent to the synchronous message passing system in which a message adversary chooses
a set $S^m$ of
$m$ processors at a round and removes some of the messages they send. But, nevertheless, this removal is constrained
by the condition that the adversary removes at most one of the messages sent between two members of $S^m$.
We call this system constrained-\SMPm.

Obviously an iteration of the asynchronous SWMR iterated shared memory, simulates a synchronous round of  constrained-\SMPm.

In the reverse direction we notice that constrained-\SMPm can simulate an iterated SWMR collect step in $2n-1$ rounds \cite{AG15},
since a round of  constrained-\SMPm is obviously a round of  constrained-SMP$_{n-1}$.
Hence we can now simulate the Atomic-Snapshot algorithm in \cite{Afek:1993}.
Since in each round of constrained-\SMPm processors ``read'' from at least $n-m$ processors,
the minimum size snapshot will be $n-m$.

Can we remove the constrain on constrained-\SMPm, so that the message adversary chooses $S^m$ and can possibly remove
all their messages violating the requirement that among pairs in $S^m$ at least one message survives?
We can, when all we want is to solve colorless tasks and $m<(1/2)n$.
When $m<(1/2)n$ a processor that through the rounds hears that at least
$n/2+1$ processors have heard from it, can be sure that in the next round all will, since
one of these $n/2+1$ processors will not be chosen by the adversary in the next round.
Hence we can see that at least $n-m$ processors can progress simulating read write.
What we lose is that if the adversary sticks with a fixed set $S^m$ those
processors cannot communicate. 
But at least $n-m$ processors will have outputs and $S^m$ can adopt theirs.

Our last system of interest is $\SBMPfm$. This system is like \SMPfm only that the adversary can now tamper
with messages from the fixed set $S^f$. Thus, it is, in a sense, like having $f$ fixed byzantine faults and $m$ mobile omission faults.

\section{The Lower Bound}\label{sec:lower-bound}

\begin{theorem}
For the hybrid omission $n,f,m$, if $n<f(m+1)+1$, there is no algorithm to solve ($m+1$)-vector-set-consensus.
\end{theorem}

\begin{proof}
The proof is based on a simulation that uses  constructions similar to those used in~\cite{GGP05, GGP11}.
The details appear in the references, but nevertheless we sketch the construction of the simulation.

W.l.o.g. by way of contradiction an algorithm exists with all processors sending messages to all. 
Take $m+2$ BG processors \cite{GGP05,GGP11}
 simulating one round before moving to the next. 
In each round a simulator decides by safe-agreement \cite{BG,BGLR} 
whether a message sent is received or not.
A BG processor will  claim that a message was removed if it does not
know the simulated local state of the processor that sends the message.

Since we are dealing with \SMPfm we can equivalently deal with the election-version of set-consensus.
Initially, every simulator tries to install its input value as the input to all
simulated processors. At most $m+1$ safe agreements modules may be blocked and the corresponding 
processors cannot be simulated as sending messages. Thus, to proceed, when a simulator does not
know the local state of a processor, then  it will try to reach agreement 
that the message was removed. 

But in the meantime, the safe agreement for this processor might be resolved, hence other processors
may contend that a message was sent. If we do these message delivery safe agreements for a processor 
proceeding negatively (a message was not sent) from the highest index receiver to the lowest, and
in the opposite positive direction, from low index to high, for the message sent, at most one safe-agreement about a single message
from this processor may be blocked at the index that the positives and the negatives meet.
This will manifest itself as a message to some processor that we do not know whether it was removed
or received, and therefore we do not know the local state of this potentially receiving processor.

Thus, the initial possible lack of knowledge about the $m+1$ inputs may propagate to at most $m+1$ 
omission failures from round to round. Thus we get an execution of a synchronous system $n,m+1$ 
with $n$ processors and $m+1$-omission faults.


Now we move to the ramification of this simulation.  If there exists an algorithm
for the system $n,f,m$ such that each run of $n,m$ can be viewed as a run of $n,f,m$,
then the algorithm that obtains $m+1$-vector-set consensus for $n,f,m$ will be an algorithm for $m+1$-vector-set consensus
for the system $n,m+1$, contradicting \cite{BG,HS,SZ}. 

Hence, it must be that there exits a run of  $n,m+1$ that cannot be explained as a run of
an algorithm for $m+1$-vector-set consensus of a synchronous hybrid system $n,f,m$.

The system $n,f,m$ can have $m+1$ omission failure in a round by choosing each round $m+1$ processors 
from which to remove messages as to simulate a round of $n,m+1$. Observe that at least one of these $m+1$
processor chosen at a round has to be on the account of the $f$ fixed processors, since in $n,f,m$ we have only 
$m$-mobile omission faults at a round rather than $m+1$.

Thus, we want to show that for $n$ small enough, $n<f(m+1)+1$, given a run of $n,m+1$ we can allocate $f$-fixed
processors that explain the run as a run of $n,f,m$.


We take an infinite run of $n,m+1$ and divide it into chunks of $n$ rounds. In the first 
round processors $1,2, \ldots, m+1$ omit, in the second round $2,3, \ldots m+2$ omit, etc, with wrap around
at round $n$. Then repeat the same for the second chunk etc.

Thus, if we take any $k$ chunks like this, a processor appears in all the chunks $(m+1)k$ times each at a different round.
Thus, since we have $f$ fixed faults, the largest number of rounds we can justify with these
fixed faults is $f(m+1)k$. But if $nk>f(m+1)k$ we will not be able to justify it as an $n,f,m$ run.

Obviously for any $n \leq f(m+1)$ we will be able to justify it as an $n,f,m$ run.
Our example made the repetition of each processor equal. If it is unequal we will attribute
the processors that are at the top $f$ in ranking of repetitions as the fixed set and will
able to justify any run of $n,m+1$ as a run of $n,f,m$.
\commentout{
Take some $k$ rounds such that the number of appearances of each processor omitting is the same and
the run is completely symmetric under the exchange of two ids modulo permutation of rounds.
Let $A_i$ be the set of processors assigned to round $i$.
$$A_i=\{p_\ell\mid \ell=1+ ((j-1+(i-1)(m+1))\mod n), \mbox{ for }1\le j\le m+1\}\,.$$

Thus a processor appears in at most $(m+1)k/n$ rounds. Thus, the largest number of rounds for which
we can attribute at least one id to the $f$-fixed omission set is $f(m+1)k/n$. Since we need at least
one processor from the $f$-fixed set in each round we get the condition $f(m+1)k/n \geq k,$
which yields $n \leq f(m+1)$.
}
\end{proof}
The lower bound proof implies that if we take $n\geq f(m+1)+1$ we will not be able to explain the run as an $n,m+1$ run.
What is left is to show that indeed for $n\geq f(m+1)+1$ we have an algorithm that obtains ($m+1$)-vector-set consensus.

An easy repetition of the lower bound arguments above show that if $n< \lfloor f(m+j)/j \rfloor+1$ we cannot solve $m+j$-set consensus.


\section{Binary-$k$-vector-set Consensus}
For simplicity of exposition we will use only the binary-$k$-vector-set consensus even for \SMPfm as we know we can solve
the multi-valued one using \cite{Rachid}.
To show the phenomenon of consensus power growing gradually with $n$ given fixed $f$ and $m$,
we assume for convenience  that $m+f<n/2$ and $f<n/4$.
The algorithm rely on the idea of the rotating coordinator \cite{Reischuk198523,Chandra:1991:TME:112007.112027}.
Only that in hindsight we know that each phase of the implementation hides a \ca protocol \cite{Gafni98}.
\footnote{The \ca protocol is the concept behind last two rounds  of the original
Gradecast \cite{62225} protocol. }
Hence, we first pause the presentation and show how we solve \ca in \SMPfm and then in \SBMPfm.

\subsection{Commit-Adopt implementation for \SMPfm, and \SBMPfm }

In the \ca protocol each processor invokes \ca with an initial value. Each processor, $p$,  returns as an output a pair $\tri{v_p,e_p}$, where $e_p\in\{\mbox{\small\sc commit, adopt}\}$. \ca ensures that:
\beginsmall{enumerate}
\item[\textup{ CA1:}]
If all processors invoke \ca with the same value, then every processor, $p$, returns with that value and with  $e_p =\, ${\sc commit}.
\item[\textup{ CA2:}]
If a processor, $p$, returns  with $e_p =\, ${\sc commit}, then for any processor $q$, $v_q=v_p$. 
\endsmall{enumerate}

A \ca algorithm in \cite{Gafni98} is given in a wait-free shared-memory model.
To use this algorithm in different models we either implement the shared-memory in the model, or just show an implementation that 
comply with the properties that make \ca algorithm in shared-memory work.
The properties are to have two iteration where at most one value will be observed as a proposal to commit
in the second round, and if a processor views only commit proposal in the second round, then any
other processor in the second round will observe a proposal to commit.

In \SMPfm if all processors start with the same bit, every processor will get all messages of the same bit and there will be at least $n/2 + 1$ of them.
A processor that receives all same bit, will propose in the second round to commit that bit.
Obviously, since majority sent same bit, no other processor will propose to commit a different bit.
 
In the second round, a processor commits if it obtains at least $n/2+ 1$ proposals to commit.
A processor that does not propose to commit does not send a message.
Obviously if one processor receives  at least $n/2 +1$ proposals to commit, others cannot miss all these proposals and will see at least one proposal to commit,
and hence will adopt that bit.

In \SBMPfm we have to worry about the $f$ processors  whose messages the adversary is allowed to tamper with.
Thus, we cannot require to see all messages of the same bit, since then the adversary can prevent commit in the case
that all started with the same bit. Nevertheless we know that if all started with the same bit then a processor will get 
at least $n/2+1$ of that bit, and at most $f$ of the complement bit.

Hence, we use this test in order to propose commit at the second round. Since in the worst case the complement bit
was send by ``correct'' processors, we nevertheless are left with more than quarter of correct processors whose input is that bit.
Hence, this set of processors in the first round will prevent any other processor to propose to commit the complement bit, 
as the size of the set is greater than $f$. A processor that does not propose to commit does not send anything in the second round.

In the second round for a processor to ``read'' a commit proposal it is enough if it obtains at least $n/4$ commit proposals of the same
bit. To commit, a processor needs to receive again at least $n/2 +1$ proposals of commit of the same bit (We can now ignore commit
of the different bit since there can be anyway at most $f$ of them). 
Again we can argue that in the second round, if any processors commits, 
we must have at least $n/4$ correct processors which send commit proposal of that bit and hence all will at least adopt that bit.

\figureref{figure:ca-omm} and \figureref{figure:ca} present both versions of the \ca protocol. 

\begin{algorithm}[!ht]
\footnotesize
\SetNlSty{textbf}{}{:}
 \setcounter{AlgoLine}{0}
\begin{tabular}{ r l }
\lnl{line:o-vote1} &  {\bf round 1:}   send $v_p$ to all;\hspace{2.2in}\hfill\textit{/* executed by processor $p$ */}\\
\\
\lnl{line:o-test}&  {\bf round 2:}  {\bf if}  all ``bits'' received are the same  {\bf then} {\bf send} the ``bit'' to all  {\bf else} {\bf do not send};\\
\nl & \tb\due \textit{/* Notations */}\\
\nl & \tb\due {\bf let} $maj$ be the bit received and let $\#maj$ be the number of processors that sent it;\\
\\
\nl &  \textit{/* Scoring */}\\
\lnl{line:o-commit}& \tb {\bf if}  $\#maj> n/2 $ {\bf then} {\bf set}  $e_p := ${\sc commit}  {\bf else} {\bf set}  $e_p := ${\sc adopt};\\
\lnl{line:o-value}& \tb {\bf if} $\#maj>0$ {\bf then} {\bf set}  $v_p :=maj$;\hfill\textit{/* otherwise remain with the original $v_p$ */}\\
\\
\nl & {\bf return} $\tri{v_p,e_p}$.
\end{tabular}
\caption{ \ca\!$(v_p)$: The Commit Adopt protocol for omission faults}\label{figure:ca-omm}
\end{algorithm}


\begin{algorithm}[!ht]
\footnotesize
\SetNlSty{textbf}{}{:}
 \setcounter{AlgoLine}{0}
\begin{tabular}{ r l }
\lnl{line:vote1} &  {\bf round 1:}   send $v_p$ to all;\hspace{2.2in}\hfill\textit{/* executed by processor $p$ */}\\
\nl & \tb\due \textit{/* Notations */}\\
\lnl{line:1st}& \tb\due {\bf let} $maj$ be the bit received the most and let$\#maj$ be the number of processors that sent it;\\
\nl & \tb\due {\bf let} $min$ be the bit received the least and let its number be $\#min$;\\
\\
\lnl{line:test}&  {\bf round 2:}  {\bf if} $\#min \leq f$ and $\#maj> n/2 $ {\bf then} {\bf send}  $maj$  {\bf else} {\bf do not send};\\
\nl & \tb\due \textit{/* Notations */}\\
\nl & \tb\due {\bf let} $maj$ be the bit received the most and let $\#maj$ be the number of processors that sent it;\\
\\
\nl &  \textit{/* Scoring */}\\
\lnl{line:commit}& \tb {\bf if}  $\#maj> n/2 $ {\bf then} {\bf set}  $e_p := ${\sc commit}  {\bf else} {\bf set}  $e_p := ${\sc adopt};\\
\lnl{line:value}& \tb {\bf if} $\#maj \geq n/4$ {\bf then} {\bf set}  $v_p :=maj$;\hfill\textit{/* otherwise remain with the original $v_p$ */}\\
\\
\nl & {\bf return} $\tri{v_p,e_p}$.
\end{tabular}
 \caption{ \ca\!$(v_p)$: The Commit Adopt protocol for Byzantine faults}\label{figure:ca}
\end{algorithm}

\subsection{Binary-$k$-vector-set Consensus Protocol}

%
%
%

We first focus on the case of $k=m+1.$  For the discussion below assume that $n=f(m+1)+1.$ 
The idea of the protocol is to run in parallel the basic process for each of the $m+1$ entries in the binary-$(m+1)$-vector-set consensus.  The process below will ensure that in each entry in the output vector different processors never produce conflicting outputs, and that for at least one entry all processors report an output.

We assign $m+1$ coordinators to each phase of the protocol, one per entry in the vector.  The coordinators play a role in a specific round of sending messages in each phase, as described below.  We run the protocol for $f(m+1)+2$ phases, each takes three rounds of message exchange. 

For a given entry all processors repeatedly exchange their values in each phase.  Each phase begins with concurrently running a \ca on the current values of all processors.  In the first phase processors use their initial input values, and later phases the values  computed by the end of the previous phase. 

Following the \ca step the coordinator of the current phase broadcasts the value it obtained from the recent \ca.  

A processor that completed the recent \ca with {\sc commit} ignores the coordinator's message and updates its value to be the committed value of the \ca.
A processor that did not complete the recent \ca with {\sc commit} adopts the value it receives from the coordinator, if it received a value, if no value was received it remains with its original value.

By the end of this value updating we are guaranteed that if the coordinator was correct when it sent its coordinator's value, then all processors will end up holding identical values.  
The reason is 
that the \ca properties imply that if a processor returns from the \ca with {\sc commit}, all processors  return from the \ca with identical values, so this is also the value the correct coordinator sends.  
If this is not the case, every processor adopts the coordinator's value, and again they hold identical values.

Observe that our assumptions are that all processors receive the values from all correct processors, even when the adversary chooses to change their messages. Therefore, the current coordinator received the correct value from the \ca as every other processor.

Once all processors hold identical values, in all future phases the \ca at each processor will return {\sc commit} with that value, no matter who the rest of the coordinators are.

The above basic process is repeated for $f(m+1)+2$ phases for all the $(m+1)$ entries of the $(m+1)$-vector-set consensus.
After the end of the last phase each processor reports output for every entry in which the latest \ca returned  {\sc commit}. 
The \ca properties imply that there will not be any conflict on output values in any index. 
Moreover, for each entry for which in one of the first $f(m+1)+1$ phases there happened to be a correct coordinator sending its value,  all processors return that value for that entry.

What we are left  to discuss is why there would always be at least one correct coordinator in at least one entry in at least one phase. 
Although this argument is repetition of the argument in the lower-bound section, we repeat it here.  Observe that we assign to each phase   $m+1$ different coordinators.
The assignment of  coordinators to phases is such that for $n=f(m+1)+1$ each one appears in exactly $f+1$ different phases.  This implies that there can be at most $f(m+1)$ phases in which at least one of the coordinators assigned to entries in that phase is from the fixed set $f$.  Look at a phase in which no coordinator is from the fixed set.  The adversary can drop messages from at most $m$ of the coordinators that send their coordinators' values in that phase.  Therefore, there should be an entry at which the coordinator sending the coordinator's value is correct.

\begin{algorithm}[!ht]
\footnotesize
\SetNlSty{textbf}{}{:}
 \setcounter{AlgoLine}{0}
\begin{tabular}{ r l }
\nl & \textit{/* Initialization */}\mbox{\hspace{2.5in}}\hfill\textit{/* executed by processor $p$*/}\\
\lnl{line:input}& {\bf let} $v(j)$ be the initial input to consensus index $j$, $1\le j\le m+1$;\hfill\textit{/* the input values*/}\\
\nl &\tb\textit{/* the permutation over the set of  $n$ processors */}\\
\lnl{line:assign} & {\bf let} $s_{i,j}=p_\ell$, where $\ell=(i+j-1)\mod n$, for $1\le i\le  f(m+1)+2$, $1\le j\le m+1$;\\
\\
\nl & \textit{/* Main loop for each of the $k$ indices, $1\le j\le m+1$, all of them in parallel */}\\
\nl & {\bf for} phase $i:=1$ to $ f(m+1)+2$ {\bf do}\\

\lnl{line:ca} & \tb   $\tri{\hat v(j),\hat e(j)}=${\bf \ca}$(v(j))$;\\
\\
\lnl{line:sender} & \tb  {\bf if} $p=s_{i,j}$ {\bf then}    send $\hat v(j)$ to all;\hfill\textit{/* the rotating coordinator for index $j$ sends its value */}\\
\lnl{line:vote} & \tb\tb  {\bf let} $v'(j)$ be the value received from $s_{i,j}$;\hfill\textit{/* $\bot$ if no value was received */}\\
\\
\lnl{line:update1} & \tb {\bf if} $\hat e(j)= \mbox{\sc commit}$ {\bf then}  $v(j) := \hat v(j)$\\
\lnl{line:update2} & \tb\tb  {\bf else if} $v'(j)\not=\bot$ {\bf then} {\bf set} $v(j) := v'(j)$;\hfill\textit{/*  adopt coordinator $s_{i,j}$   value*/}\\
\nl & {\bf end for}\\
\\
\lnl{line:output}& {\bf for each} $j$, $1\le j\le m+1$: {\bf if} $\hat e(j)= \mbox{\sc commit}$ {\bf then} $o_p(j):=v(j)$  {\bf else} $o_p(j):=\bot$;\\
\\
\nl & {\bf return} $o_p[1..m+1]$.
\end{tabular}
 \caption[caption]{$\mSetCons$:  \\\hspace{\textwidth}\mbox{\ \hspace{0.75in}}a  \Byzantine $m+1$-vector-set consensus algorithm}\label{figure:mSetCons}
\end{algorithm}

Observe that for binary values one can replace the condition in \lref{line:value} of \figureref{figure:ca} to $\#maj>0$,  since if  no processor returns with {\sc commit} then  non-Byzantine processors have sent both $0$ and $1$.  
For non-binary values, instead of testing for $\#min$ we need to test for non-$max$ values,
and can replace the condition in \lref{line:value} of \figureref{figure:ca} to $\#maj>f$.
Moreover, 
one can  add a filtering  in \lref{line:1st}  to filter out values that do not conform with what one expects to receive, since they are clearly being sent by Byzantine processors.
Similar filter can be used in \lref{line:vote} of \figureref{figure:mSetCons}. 
We do not have any use for such a filtering in the protocols of the current paper.

One can generalize the lower bound proof of \sectionref{sec:lower-bound} for
$m+k$-vector-set consensus algorithm for $k \ge 1$ to obtain a lower bound of $n> \lfloor{f(m+k)\over k}\rfloor$.  
The \mSetCons protocol of
\figureref{figure:mSetCons} can be changed accordingly and will run in $ \lfloor{f(m+k)\over k}\rfloor+2$ phases. 
When increasing the number of of processors by $m$ one can device a protocol that runs 
for only $ \lfloor{f\over (k-m)}\rfloor +2$ phase.  Thus for $n \geq f(m+1)+1$ we can solve $m+1$-set consensus.
for $n \geq f(m+2)/2 +1$ we can solve $m+2$ set consensus, etc. 


All the formal proofs appear in the appendix.
\section{Conclusions}

We introduced a new type of distributed-system call Hybrid-Message-Adversary.
It gives rise to phenomenon never seen before of set-consensus power changing gradually
even though the various types and number of faults do not change.
%
%
%
In our mind the only notion of set consensus that makes sense
in the Byzantine setting is that of binary-vector-set-consensus.
To our knowledge we are the first to ask this question, and in fact we
are still at loss but not far, we suspect, from an answer.

Next, we can imagine message adversary with mobile Byzantine faults
and combinations thereof with omission fixed or mobile faults etc..
In fact, the analogue of message adversary with mobile Byzantine faults
was studied in the domain of Cryptography under the name of mobile viruses, transient or proactive faults 
\cite{Ostrovsky:1991,Canetti1994,Rabin1998,Castro:2000},
but none looked at the relative power of tasks, let alone the set-consensus power.

Why should we? We recently \cite{CSC} 
started to suspect that ``natural systems'' can be characterized 
by their set consensus power. Thus if this is proved and we equate the set-consensus power of synchronous Byzantine
of $f$ faults and \SBMPfm with $m=f$, then they will be equivalent.


\newpage

\thispagestyle{empty}
\bibliographystyle{plain}
\bibliography{bibliography}
\newpage
\appendix

\noindent {\Large \bf Appendix}

\section{Proofs}\label{sec:Proofs}
We prove the correctness of \ca for the \SBMPfm model. The proof for the \SMPfm is much simpler.

\begin{claim}\label{claim:CA}
Assume $n>\max(2f+2m,4f)$ and \SBMPfm adversary.  The protocol of \figureref{figure:ca} meets the \ca requirements.
\end{claim}
\begin{proof}
To prove property CA1 observe that if all processors send an identical value in \lref{line:vote1} of \figureref{figure:ca} then each processor can receive at most  $f$ different values, from processors the adversary tampers with their messages. Therefore, at every processor the test in \lref{line:test} will bring it to propose this identical value.  Exactly for the same reasons each processor, $p$, will set up $e_p$ to be {\sc commit} and will return the identical value.

For CA2, assume that  processor $p$ proposes a value in the second round of the protocol. Thus, the test in \lref{line:test} is $\true$.
Denote by  $maj_p$,  $\#maj_p$ and  $\hat n_p$ the parameters processor $p$ used in \lref{line:test} and $V_p$ the multiset of $\bar n_p$ values it received, and 
respectively for another arbitrary processor $q$.  
Consider two cases. The first case is when $m\ge f$. 
In our model $V_p$ may not contain at most the $m$ values from omission faults and some of the values from the Byzantine processors. Since $\#maj_p> n/2 $ and $\#maj_p\ge  \hat n_p-f$ there are at least $m+1$ 
values in $V_p$ from  processors that are correct at the current sending step  and sent $v_p$ to all.  $V_q$ should contain all these processors' values. The assumption that $m\ge t$ implies that $V_q$ contains at least $f+1$ value of $v_p$, therefore, \lref{line:test} of \figureref{figure:ca} implies that if  decides to propose a value it should be that $v_q=v_p$.  Now consider the case that $m<f$.  In this case the condition $n>4f$ implies that $n/2>2f.$ This implies that  $V_p$ contains at least $n/2-f$  values from  processors that are correct at the current sending step that sent $v_p$, thus at least $f+1$, and the rest of the above arguments hold.

The above argument implies that no two non-Byzantine processors send different values in the second round. 
Assume that  processor $p$ commits in \lref{line:commit}.
Let $v$ be a value committed to.  
Assume the above notations for the messages received in the 2nd round.
Neither $V_q$ nor $V_p$ can contain more than $f$ non $\bot$ values the are not $v$.
Thus, the protocol implies that $maj_p=v$.
By definition we know that $|V_q\cap V_p|\ge n/2$, all of which are correct at the current message sending step. 
Since $\#maj_p\ge  \hat n_p-f$ and $\#maj_p> n/2 $,  $V_q$ contains at least $n/2-f$ copies of $v$, thus more than $n/4$ and more than $f$.  
\end{proof}

\begin{claim}\label{claim:samev}
Assume $n>\max(2f+2m,4f)$ and \SBMPfm adversary.  In the \mSetCons protocol, if all processors' initial values to a given index $j$ are the same, they output that value for that index at the end.
\end{claim}
\begin{proof}
In \mSetCons processors update their initial values either in \lref{line:update1} or \lref{line:update2} of \figureref{figure:mSetCons}.  By \claimref{claim:CA} we know that any invocation of it will result with all processors obtaining the identical value and with evaluation ``commit", therefore \lref{line:update2} will never be executed, and at the end of the protocol all will produce an identical value.
\end{proof}

\begin{claim}\label{claim:correct-coordinator}
Assume $n>\max(2f+2m,4f)$ and \SBMPfm adversary.  In the \mSetCons protocol, if  for some index $j$ at some phase the coordinator $p$ is correct when executing \lref{line:sender}, then after executing \lref{line:update1} or \lref{line:update2} of \figureref{figure:mSetCons} in that phase, $v(j)$ is identical at all processors. 
\end{claim}
\begin{proof}
Let $p$ be the correct processor executing \lref{line:sender} of some index $j$, in some phase $r$. Consider two cases.  If there is any processor that completes  the \ca in \lref{line:ca} of phase $r$ with ``commit" and the other if none.  In the first case, by the \ca properties we know that all processors complete the \ca with the same value (including the faulty processors).  Therefore, the value $p$ will send in \lref{line:sender} is the same, and therefore all processors will complete \lref{line:vote} with the same value. This implies that no matter which of the two lines, \lref{line:update1} or \lref{line:update2}, any processor executes, all obtain the same value.
In the second case, no processor completes \lref{line:ca} with ``commit", and therefore all will execute \lref{line:update2} and will obtain the value the correct sender sent when it was correct while executing  \lref{line:sender}.
\end{proof}

%
%
%
%
%

%
\begin{theorem}\label{thm:mSetCons}
If $n>\max\bigl(f(m+1), 4f, 2(m+f)\bigr)$ and assuming \SBMPfm adversary, the \mSetCons protocol satisfies the properties of $k$-vector-set consensus, for $k= m+1$.
\end{theorem}
\begin{proof}
In order to prevent having any correct coordinator in a phase all $m+1$ processors assigned to be coordinators in that phase need to either be in the fixed set of $f$ faulty, or one of the $m$ processors that suffer from omission in that phase.
The definition of $s_{i,j}=p_\ell$, where $\ell=(i+j-1)\mod n$, for $1\le i\le  f(m+1)+2$, $1\le j\le m+1,$ assigns each processor $p$ to exactly $m+1$ times in the first $f(m+1)+1$ phases.  The fix set of $f$ processors appear in at most $f(m+1)$ phases.  Therefore, there is a phase in which none of them appear.  Since omission faults can silence at most $m$ processors in that phase, there is a correct processor executing \lref{line:sender}  of \figureref{figure:mSetCons}  in that phase. Let $j$ be the index of that processor.

By \claimref{claim:correct-coordinator} we know that by the end of  that phase all processors hold the same values in their $v(j)$.  From the next phase on, until the end of phase $f(m+1)+2$ all will complete \lref{line:ca} in \figureref{figure:mSetCons} with ``commit", and will end up having the same value in the $j$-th index.  Moreover,  for any other index $j'$, two processors that  assign a value to that index, assign the same value, since it is the value they completes \lref{line:ca} of index $j'$ with ``commit", and it is an identical value.
The remaining property of \kSetCons obviously holds as well.
\end{proof}
%


%
%

\end{document}